\documentclass[11pt,twocolumn]{article}

\usepackage{booktabs} 

\usepackage[utf8]{inputenc}

\usepackage{authblk}

\usepackage{amsmath,amssymb,amsfonts,amsthm}

\usepackage{mathtools}

\usepackage{graphicx}
\usepackage{epstopdf}

\usepackage{url}

\usepackage{algorithm}
\usepackage{xcolor}
\usepackage{listings}
\usepackage[export]{adjustbox}

\newcounter{cntTheorem}
\newtheorem{definition}[cntTheorem]{Definition}
\newtheorem{lemma}[cntTheorem]{Lemma}
\newtheorem{proposition}[cntTheorem]{Proposition}
\newtheorem{corollary}[cntTheorem]{Corollary}

\newcommand{\dk}{\ensuremath{\texttt{DK}}{}}

\newcommand{\sdk}{\ensuremath{\texttt{SDK}}{}}
\newcommand{\shd}{\ensuremath{\texttt{SHD}}{}}

\newcommand{\smallerthan}{\ensuremath{\sqsubseteq}}
\newcommand{\samesize}{\ensuremath{\equiv}}

\renewcommand{\geq}{\geqslant}
\renewcommand{\leq}{\leqslant}

\renewcommand{\phi}{\ensuremath{\varphi}}
\renewcommand{\epsilon}{\ensuremath{\varepsilon}}

\newcount\colveccount
\newcommand*\colvec[1]{
        \global\colveccount#1
        \begin{pmatrix}
        \colvecnext
}
\def\colvecnext#1{
        #1
        \global\advance\colveccount-1
        \ifnum\colveccount>0
                \\
                \expandafter\colvecnext
        \else
                \end{pmatrix}
        \fi
}

\newcommand{\deff}{\coloneqq}

\newcommand{\NN}{\ensuremath{\mathbb{N}}}
\newcommand{\RR}{\ensuremath{\mathbb{R}}}
 
\newcommand{\Mm}{\ensuremath{\mathcal{M}}}

\begin{document}

\title{The SuperM-Tree: Indexing metric spaces with sized objects}

\author[1]{J\"org P. Bachmann}
\affil[1]{ \texttt{joerg.bachmann@informatik.hu-berlin.de}}
\date{\today}

\maketitle

\begin{abstract}
%
%
    A common approach to implementing similarity search applications is the usage of distance functions, where small distances indicate high similarity.
    In the case of metric distance functions, metric index structures can be used to accelerate nearest neighbor queries.
    On the other hand, many applications ask for approximate subsequences or subsets, e.\,g. searching for a similar partial sequence of a gene, for a similar scene in a movie, or for a similar object in a picture which is represented by a set of multidimensional features.
    Metric index structures such as the M-Tree cannot be utilized for these tasks because of the symmetry of the metric distance functions.
    
    In this work, we propose the SuperM-Tree as an extension of the M-Tree where approximate subsequence and subset queries become nearest neighbor queries.
    In order to do this, we introduce metric subset spaces as a generalized concept of metric spaces.
    Various metric distance functions can be extended to metric subset distance functions, e.\,g. the Euclidean distance (on windows), the Hausdorff distance (on subsets), the Edit distance and the Dog-Keeper distance (on subsequences).
    We show that these examples subsume the applications mentioned above.
\end{abstract}

\section{Introduction}
\label{sec:introduction}

The most common index structure used in relational database management systems (RDBMSs) is the B$^+$-Tree \cite{BPlusTree}, which provides fast range queries on linearly ordered data.
The R$^*$-Tree \cite{RStarTree}
is the usual index structure used to index multi dimensional data and offers the capability for fast multi dimensional range queries.

Numerous specific index structures exist to offer more expressive queries on more complex data types.
This includes index structures for sets and set containment joins \cite{PRETTI,PIEJoin,SetContainmentJoin}, for strings and similarity search regarding the Edit distance \cite{RCSI}, and for nearest neighbor queries in any metric space \cite{mtree,mvptree,covertree,mindex}.

On the other hand, metric index structures are more generic and support range queries on a wide range of different complex data types.
Examples are the Euclidean distance on sequences of the same length, the Edit distance (ED) and the Dog-Keeper distance (DK) on sequences or multi-dimensional trajectories of arbitrary lengths \cite{WDK,ImprovedDK}, and the Hausdorff distance on sets (e.\,g. on sets of integers and sets of feature vectors).
However, the symmetry of a metric strongly constrains the expressiveness of the queries and prevents containment queries which ask for subsequences, subsets, and the like.

In order to support approximate subset queries, we need to disregard the symmetry of the metric axioms and we need a relation regarding the size of the objects.
Hence, we introduce the \emph{metric subset space} $(\Mm,d,\smallerthan)$ consisting of a set of objects $\Mm$, an (asym-) metric function $d$ and a total preorder $\smallerthan$ ordering the objects by their size.
We also reduce the necessity of the triangle inequality, such that the triangle inequality $d(x,z)\leq d(x,y)+d(y,z)$ only needs to hold for objects $x,y,z\in\Mm$ with $x\smallerthan y\smallerthan z$.

Based on the M-Tree, we propose the SuperM-Tree which indexes metric subset spaces.
The SuperM-Tree supports $k$ nearest neighbor and $\varepsilon$ nearest neighbor queries:
Given a query object $q$, the result of a \emph{sub query} consists of objects $p$, such that $p\smallerthan q$ ($p$ is not larger than $q$) and $p$ is similar to a part of $q$ (regarding the metric subset distance function).
As a demonstration of the generality of metric subset spaces, we provide several examples including approximate set containment queries and subsequence queries with various distance functions, such as the Euclidean distance (on windows), the Hausdorff distance (on subsets), and the Dog-Keeper distance (on subsequences).

In summary, we propose a general index structure for database systems, which supports natural types of queries for a wide range of multimedia data types.

The rest of the paper is structured as follows:
We describe the concept of subset metrics in Section~\ref{sec:asymmetricspaces} including examples.
Section~\ref{sec:datastructure} introduces the datastructure of the SuperM-Tree and the algorithms for insertion and searching of objects.
In Section~\ref{sec:evaluation} we evaluate the efficiency of the SuperM-Tree on three different subset distance functions and show that the speedup against a linear scan search algorithm increases with increasing size of the datasets.

\subsection{Notation}

A defines B is denoted by $A\deff B$ or $A:\Leftrightarrow B$.
For finite sequences $T$, $|T|$ denotes the number of elements in that sequence and $T_i$ denotes the $i$-th element in the sequence.
Subsequences of $T$ starting at $i$ with a length of $\ell$ are denoted by $T_i^\ell$.
Positive integers including zero are denoted by $\NN$ and positive real numbers including zero are denoted by $\RR^{\geq 0}$.

\section{Metric Subset Spaces}
\label{sec:asymmetricspaces}

Consider a set of objects $\Mm$, for example sequences of arbitrary lengths or sets of $k$-dimensional vectors.
Natural queries put objects in a certain ``containment'' relationship when they ask for subsets or subsequences.
We describe this containment relationship by a total preorder (denoted with '$\smallerthan$').
Intuitively, this relationship requires for each pair of objects $x$ and $y$, that $x$ is either not larger, not smaller, or neither larger nor smaller than $y$.

We also relax the triangle inequality by only asking for it on transitive chains, i.\,e. $d(x,z)\leq d(x,y)+d(y,z)$ if $x\smallerthan y\smallerthan z$.

\begin{definition}[Total Preorder]
    A total preorder on $\Mm$ is a relation $\smallerthan$ on $\Mm\times\Mm$, such that
    \begin{align*}
        \forall x\in\Mm:& x\smallerthan x \\
        \forall x,y\in\Mm:& x\smallerthan y \lor y\smallerthan x\\
        \forall x,y,z\in\Mm:& x\smallerthan y \land y\smallerthan z \longrightarrow x\smallerthan z
    \end{align*}
\end{definition}

Note that a total preorder defines an equivalence relation $\samesize$ via
\begin{align*}
    x\samesize y\Longleftrightarrow x\smallerthan y\land y\smallerthan x
\end{align*}

Examples for total preorders include the comparison of the length of time series or the cardinality of finite sets.

\begin{definition}[Metric Subset Space]
    A metric subset space is a 3-tuple $(\Mm, \smallerthan, d)$ consisting of a set $\Mm$, a total preorder $\smallerthan$ on $\Mm$, and a function $d:\Mm\times\Mm\rightarrow\RR^{\geq 0}$ which satisfies the following axioms:
    \begin{align*}
        & S_1)& \forall x,y,z\in\Mm:\,   & x\smallerthan y\smallerthan z \longrightarrow \\
        &     &                          & d(x,z)\leq d(x,y)+d(y,z) \\
        & S_2)& \forall x,y\in\Mm:\,     & d(x,y)=0 \longleftrightarrow x=y
    \end{align*}
\end{definition}
Compared to metric spaces, $S_1$ relaxes the triangle inequality and the symmetry is missing completely.
Hence, metric subset spaces generalize metric spaces in that each metric space with a total preorder also is a metric subset space.

The following Corollary~\ref{cor:reverse} claims, that we can switch the parameters for the total preorder as well as the distance function and we still have a metric subset space.
Intuitively spoken, this converts subset ordering to superset ordering and vice versa.
\begin{corollary}
    \label{cor:reverse}
    Let $(\Mm, \smallerthan, d)$ be a metric subset space, $y\sqsupseteq x :\Longleftrightarrow x\smallerthan y$, and $\tilde d(y,x) \coloneqq d(x,y)$.
    Then, $(\Mm, \sqsupseteq, \tilde d)$ is a metric subset space.
\end{corollary}

The following three common examples show the generality of this approach.

\subsection{L2 distance on subsequences}
\label{sec:l2subseq}

Let $\Mm$ be a set of real valued sequences (i.\,e. time series) and $S,T\in\Mm$ be such two sequences.
The canonnical total preorder for subsequence search is the ``shorter or equal than'' relation, i.\,e. $S\smallerthan T :\Longleftrightarrow |S|\leq|T|$.
For sequences with $S\smallerthan T$, a common approach for subsequence search is the windowing approach with the Euclidean distance:
\begin{align*}
    d_2(S,T) \deff \min_{j} \sqrt{\sum_i \left( S_i-T_{i+j} \right)^{2}}.
\end{align*}

We show that this model for subsequence distances yields a metric subset space.
\begin{proposition}
    \label{prop:l2metricsubsetspace}
    $(\Mm,\smallerthan,d_2)$ is a metric subset space.
\end{proposition}
The proposition even holds for subsequences of elements from another metric space (e.\,g. $n$-dimensional vectors).

\begin{proof}
    Consider three sequences $S,T,U\in\Mm$ with $S\smallerthan T\smallerthan U$ and fix $s,t\in\NN$ such that
    \begin{align*}
        d_2(S,T) &= d_2\left( S,T_{s}^{|S|} \right) \\
        d_2(T,U) &= d_2\left( T,U_{t}^{|T|} \right)
    \end{align*}
    \begin{figure}
        \begin{center}
            \includegraphics[width=.7\linewidth]{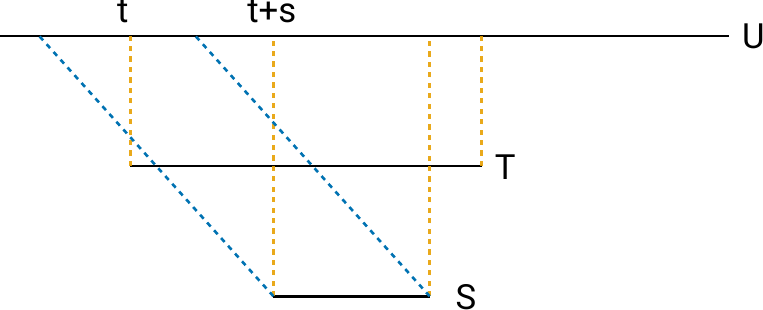}
        \end{center}
        \caption{Sketch for variables choice in proof of Proposition~\ref{prop:l2metricsubsetspace}.
                 The striped lines indicate the best matches of the subsequences.}
        \label{fig:l2metricsubsetspace}
    \end{figure}
    In order to prove the triangle inequality, we need three intermediate inequalities.
    \paragraph*{1:}
    Since $d_2$ is a metric on common length sequences, the following triangle inequality holds:
    \begin{align}
        d_2(S,U_{t+s}^{|S|}) &\leq d_2\left( S,T_{s}^{|S|} \right) + d_2\left( T_{s}^{|S|}, U_{t+s}^{|S|} \right)
        \label{eq:l2triangle_1}
    \end{align}
    \paragraph*{2:}
    The monotonicity of $d_2$ regarding the length yields the second inequality:
    \begin{align}
        d_2\left( T_s^{|S|}, U_{t+s}^{|S|} \right) &= \sum_{i=s}^{s+|S|-1} \left(T_i,U_{t+i}\right)^2 \\
        &\leq \sum_{i=0}^{|T|-1} \left(T_i,U_{t+i}\right)^2 \\
        &= d_2\left( T, U_t^{|T|} \right)
        \label{eq:l2triangle_2}
    \end{align}
    \paragraph*{3:}
    The last inequality uses the properties of the windowing approach, which is looking for the best match:
    \begin{align}
        d_2\left( S,U \right) = \min_{j} d_2\left( S, U_j^{|S|} \right) \leq d_2\left( S, U_{t+s}^{|S|} \right)
        \label{eq:l2triangle_3}
    \end{align}
    Combining inequalities~\eqref{eq:l2triangle_1},\eqref{eq:l2triangle_2},~and~\eqref{eq:l2triangle_3} proves the proposition:
    \begin{align*}
        d_2(S,U) &\leq d_2\left( S, U_{t+s}^{|S|} \right) \\
        &\leq d_2\left( S,T_{s}^{|S|} \right) + d_2\left( T_{s}^{|S|}, U_{t+s}^{|S|} \right) \\
        &\leq d_2\left( S,T_{s}^{|S|} \right) + d_2\left( T, U_t^{|T|} \right) \\
        &= d_2(S,T) + d_2(T,U)
    \end{align*}
\end{proof}

\subsection{Dog-Keeper distance on subsequences}
\label{sec:dksubseq}

Again, let $\Mm$ be a set of real valued sequences, $S,T\in\Mm$, and $\smallerthan$ the length comparison.
We consider the Dog-Keeper distance (\dk):
\begin{align*}
    \dk(S, \emptyset) &\deff \infty \quad \dk( \emptyset, T)\deff \infty \quad \dk( \emptyset, \emptyset ) \deff 0\\
    \dk(S,T) &\deff \\
    \min &\begin{cases}
        \max\left\{ |S_1-T_1|, \dk\left( S_1^{|S|-1}, T_1^{|T|-1} \right) \right\} \\
        \max\left\{ |S_1-T_1|, \dk\left( S_0^{|S|}, T_1^{|T|-1} \right) \right\} \\
        \max\left\{ |S_1-T_1|, \dk\left( S_1^{|S|-1}, T_0^{|T|} \right) \right\}
    \end{cases}
\end{align*}
Although this function is defined recursively, common algorithms use dynamic programming to compute the distance with quadratic complexity \cite{computingfrechet}.

Since $\dk$ is able to stretch the ``time'' axis, we generalize the windowing approach to windows of arbitrary length in order to support subsequence comparison:
\begin{align*}
    \sdk(S,T) &\deff \min_{j,\ell} \dk\left( S, T_j^\ell \right)
\end{align*}
The difference to the windowing approach from Section~\ref{sec:l2subseq} is that the length of the subsequence in $T$ is not bound to the length $|S|$.

Similar to the windowed Euclidean distance, $\sdk$ yields a metric subset space.
\begin{proposition}
    \label{prop:dkmetricsubsetspace}
    $(\Mm,\smallerthan,\sdk)$ is a metric subset space.
\end{proposition}

\begin{proof}
    Similar to Proposition~\ref{prop:l2metricsubsetspace}, it suffices to prove the triangle inequality for arbitrary $S,T,U\in\Mm$ with $S\smallerthan T\smallerthan U$.
    We fix $s_a,s_\ell,t_a,t_\ell\in\NN$ such that
    \begin{align*}
        \sdk(T,U) &= \dk\left( T,U_{t_a}^{t_\ell} \right) \\
        \sdk(S,T) &= \dk\left( S,T_{s_a}^{s_\ell} \right)
    \end{align*}
    \begin{figure}
        \begin{center}
            \includegraphics[width=.7\linewidth]{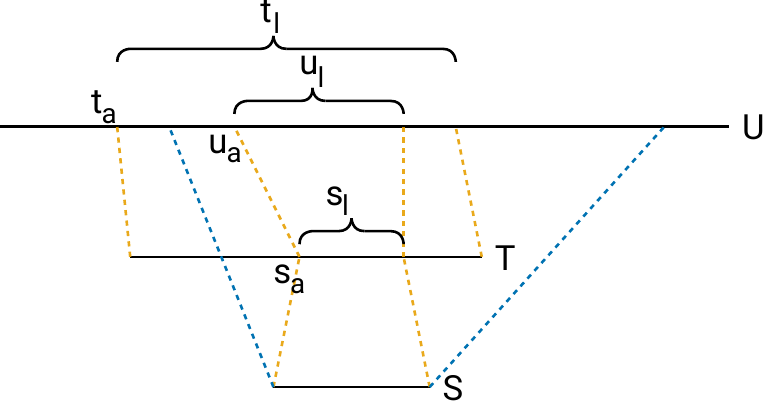}
        \end{center}
        \caption{Sketch for variables choice in proof of Proposition~\ref{prop:dkmetricsubsetspace}.
                 The striped lines indicate the best matches of the subsequences.}
        \label{fig:dkmetricsubsetspace}
    \end{figure}
    For the following thoughts, remember that the computation of $\dk\left( T,U_{t_a}^{t_\ell} \right)$ maps elements of one $T$ to elements of $U$ and thus each subsequence of $T$ to a certain subsequence of $U$. 
    Fixiating $u_a,u_\ell$ such that the subsequence $T_{s_a}^{s_\ell}$ is mapped to $U_{u_a}^{u_\ell}$ implies $\dk\left( T_{s_a}^{s_\ell}, U_{u_a}^{u_\ell} \right) \leq \dk\left( T, U_{t_a}^{t_\ell} \right)$.
    For these subsequences, the following inequality holds since $\dk$ is a metric distance function \cite{WDK}.
    \begin{align}
        \dk\left(S,U_{u_a}^{u_\ell}\right) &\leq \dk\left( S,T_{s_a}^{s_\ell} \right) + \dk\left( T_{s_a}^{s_\ell}, U_{u_a}^{u_\ell} \right)
        \label{eq:dktriangle_1}
    \end{align}
    The first summand of the right hand side of Inequality~\ref{eq:dktriangle_1} equals to $\sdk(S,T)$ by choice of $s_a$ and $s_\ell$.
    As pointed out before, the second summand is a lower bound to $\dk\left( T, U_{t_a}^{t_\ell} \right)$.
    Hence, the triangle inequality for $\sdk$ holds:
    \begin{align*}
        \sdk(S,U) &\leq \dk\left( S, U_{u_a}^{u_\ell} \right) \\
        &\leq \dk\left( S, T_{s_a}^{s_\ell} \right) + \dk\left( T_{s_a}^{s_\ell}, U_{u_a}^{u_\ell} \right) \\
        &\leq \dk\left( S, T_{s_a}^{s_\ell} \right) + \dk\left( T, U_{t_a}^{t_\ell} \right) \\
        &= \sdk\left( S, T \right) + \dk\left( T, U \right)
    \end{align*}
\end{proof}

\subsection{Hausdorff distance on subsets}
\label{sec:hdsubset}

To provide an example for metric subset spaces from another data domain, let the members of $\Mm$ be sets of real values.
Our total preorder compares the cardinalities of sets $A,B\in\Mm$, i.\,e. $A\smallerthan B :\Longleftrightarrow |A|\leq|B|$.
The Hausdorff distance is a common metric distance function on sets of the same size:
\begin{align*}
    &\textsc{Hausdorff}(A,B) \deff \\
    &\quad \max\left\{ \max_{a\in A} \min_{b\in B} |a-b|, \max_{b\in B} \min_{a\in A} |a-b| \right\}
\end{align*}
In more general, the Hausdorff is a metric for sets of arbitrary elements from any other metric space including Euclidean vector spaces.
We stick to sets of real values to keep the examples simple.

In the case of subset comparison, we don't need the symmetry of the Hausdorff distance.
Omitting the second part already yields a subset distance function ($\shd$) with similar semantics to that of Hausdorff:
\begin{align*}
    \shd(A,B) &\deff \max_{a\in A} \min_{b\in B} |a-b|
\end{align*}
While $\shd$ tries to match $A$ to a subset of $B$, the original Hausdorff distance does both directions, i.\,e. $\textsc{Hausdorff}(A,B)=\max\left\{\shd(A,B), \shd(B,A)\right\}$.
\begin{proposition}
    \label{prop:hdmetricsubsetspace}
    $(\Mm,\smallerthan,\shd)$ is a metric subset space.
\end{proposition}

\begin{proof}
    As in the proofs of Proposition~\ref{prop:l2metricsubsetspace}~and~\ref{prop:dkmetricsubsetspace}, it suffices to prove the triangle inequality for arbitrary $A,B,C\in\Mm$ with $A\smallerthan B\smallerthan C$.
    Therefore, fixiate mappings $f:A\rightarrow B$ and $g:B\rightarrow C$, such that
    \begin{align*}
        \forall a\in A: f(a) &= \min_{b\in B} |a-b| \\
        \forall b\in B: f(b) &= \min_{c\in C} |b-c|
    \end{align*}
    \begin{figure}
        \begin{center}
            \includegraphics[width=.4\linewidth]{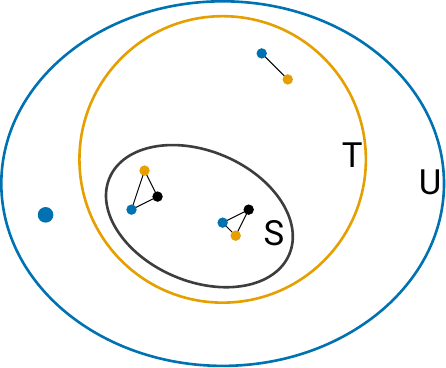}
        \end{center}
        \caption{Sketch for variables choice in proof of Proposition~\ref{prop:hdmetricsubsetspace}.
                 The striped lines indicate the best matches of the subsequences.}
        \label{fig:hdmetricsubsetspace}
    \end{figure}
    i.\,e. $f$ and $g$ map all elements to their corresponding nearest neighbor.
    Hence, $\shd(A,B) = \max_{a\in A} |a-f(a)|$ and $\shd(B,C)=\max_{b\in B} |b-g(b)|$.
    The triangle inequality is a result of the following observation:
    \begin{align*}
        \forall a\in A: \min_{c\in C} |a-c| &\leq |a-g(f(a))| \\
        \leq& |a-f(a)| + |f(a)-g(f(a))| \\
        \leq& \shd(A,B) + \shd\left( f(A), C \right) \\
        \leq& \shd(A,B) + \shd(B,C)
    \end{align*}
\end{proof}

\section{The SuperM-Tree}
\label{sec:datastructure}

The SuperM-Tree is derived from the M-Tree \cite{mtree} and differs in that it indexes metric subset spaces $(\Mm,\smallerthan,d)$ instead of metric spaces.
It is a balanced tree with nodes of a certain capacity (which can be either fixed size or variable sized as in \cite{xtree}).
Inner nodes contain \emph{routing objects} which cover an area of the metric subset space.
The covered area includes all objects in the corresponding subtree.
Leaf nodes contain the actual entries.

Additionally, the SuperM-Tree needs the objects to be ordered according to the total preorder of the metric subset space along each path from the root to the leafs.
That way, it assures that the triangle inequality holds to prune certain subtrees.
The insert and delete algorithms including the split and merge strategies need to be adapted in order to respect the ordering condition.
Due to space limitations, we do not present the algorithms for deletion.

\subsection{Structure of SuperM-Tree Nodes}

Since the structure of a SuperM-Tree is derived from the M-Tree \cite{mtree}, we use the same notation.
Leaf nodes store the indexed (key) objects, whereas internal nodes store \emph{routing objects} which help in pruning and navigating through the tree.

Each routing object $O_r$ is associated with a pointer to the root node $T(O_r)$ of its subtree, called the \emph{covered tree} of $O_r$.
Additionally to the M-Tree, the SuperM-Tree expects that $O_r\smallerthan O_j$ for each object $O_j$ in its covered tree.
The routing objects are also associated with a radius $r(O_r) \geq 0$ called the \emph{covering radius} of $O_r$, as well as the distance $d(P(O_r),O_r)$ to their parent $P(O_r)$.
All indexed objects in the covered tree of $O_r$ are within the distance $r(O_r)$ from $O_r$.

Leaf nodes only store the objects $O_j$ themselves and their distance to the parent $d(P(O_j),O_j)$.
In real world scenarios, additional information (e.\,g. data pointer or tuple identifier) can be stored per object.

\subsection{Similarity Queries}

The SuperM-Tree supports range queries and nearest neighbor queries.
Since nearest neighbor queries simply adapt the search radius analogously to the M-Tree while traversing the tree, we only discuss the range queries here.

For an object $Q\in\Mm$, the range query $NN(Q,r(Q))$ selects all database objects $O_j$ with $d(O_j,Q)\leq r(Q)$ and $O_j\smallerthan Q$.
Algorithm~\ref{alg:ennquery} provides the pseudo code for the range query.
It uses Lemma~\ref{lem:prunetriangle} and \ref{lem:prunetriangle2} to prune subtrees.

\begin{algorithm}
    \caption{Range query}
    \begin{lstlisting}
Input: node $N$, query object $Q$, search radius $r(Q)$
$R\deff\emptyset$
if $N$ is a leaf node
  for each $O_j\in N$
    if $O_j\smallerthan Q$ and $d(O_j,Q)\leq r(Q)$
      $R\deff R\cup \left\{ O_j \right\}$
  return $R$
// else
for each $O_r$ in $N$
  if not $O_r\smallerthan Q$
    skip // prune by ``size''
  if $r(Q) < d(O_p,Q) - d(O_r,O_p) - r(O_r)$
    skip // pruned using Lemma $\ref{lem:prunetriangle2}$
  if $r(Q) < d(O_r,Q)-r(O_r)$
    skip // pruned using Lemma $\ref{lem:prunetriangle}$
  $R\deff R\cup search(T(O_r), Q, r(Q))$
return $R$
    \end{lstlisting}
    \label{alg:ennquery}
\end{algorithm}

\begin{lemma}
    \label{lem:prunetriangle}
    If $d(O_r,Q) > r(Q) + r(O_r)$, then $d(O_j,Q)>r(Q)$ holds for each object $O_j$ with $O_j\smallerthan Q$ in the covered subtree of $O_r$.
\end{lemma}
Remark that, since all objects $O_j$ with $O_j\not\smallerthan Q$ are not included in the search anyways, Lemma~\ref{lem:prunetriangle} implies that the covered subtree with root $T(O_r)$ can be safely pruned from the search.

\begin{proof}
    Let $O_j$ be an arbitrary but fixed object in the covered tree of $O_r$ with $O_j\smallerthan Q$.
    The ordering $O_r \smallerthan O_j \smallerthan Q$ holds by definition of the structure of a SuperM-Tree.
    Now, by definition of metric subset spaces, the triangle inequality $d(O_r,Q)\leq d(O_r,O_j)+d(O_j,Q)$ holds.
    The structure of a SuperM-Tree requires $d(O_r,O_j)\leq r(O_r)$, thus $d(O_r,Q)\leq r(O_r)+d(O_j,Q)$, i.\,e. $d(O_j,Q) \geq d(O_r,Q)-r(O_r)$.
    Thus, if $d(O_r,Q) > r(Q)+r(O_r)$, then $d(O_j,Q) > r(Q)+r(O_r)-r(O_r) = r(Q)$.
\end{proof}

\begin{lemma}
    \label{lem:prunetriangle2}
    If $d(O_p,Q) > r(Q) + r(O_r) + d(O_p,O_r)$, then $d(O_j,Q)>r(Q)$ holds for each object $O_j$ with $O_j\smallerthan Q$ in the covered subtree of $O_r$.
\end{lemma}

\begin{proof}
    Let $O_j$ be an arbitrary but fixed object in the covered tree of $O_r$ with $O_j\smallerthan Q$.
    The ordering $O_p\smallerthan O_r\smallerthan O_j\smallerthan Q$ holds by definition of the structure of a SuperM-Tree.
    Hence, the triangle inequality $d(O_p,Q) \leq d(O_p,O_r)+d(O_r,O_j)+d(O_j,Q)$ holds.
    Since $d(O_r,O_j)\leq r(O_r)$ and $d(O_p,Q) > r(Q) + r(O_r) + d(O_p,O_r)$, we have $r(Q)+r(O_r)+d(O_p,O_r) < d(O_p,O_r)+r(O_r)+d(O_j,Q)$ and thus $r(Q)<d(O_j,Q)$.
\end{proof}

\subsection{Building the SuperM-Tree}
\label{sec:buildingmtree}

The SuperM-Tree is a generalized search tree (GiST \cite{gist}), i.\,e. algorithms for insertion and deletion of objects manage overflow and underflow of nodes using split and merge operations.

The \texttt{Insert} algorithm recursively descends the SuperM-Tree down to a leaf node in which to insert the object.
At each inner node, we follow the routing object with the smallest distance to the new object.
Two events may occur:
        1. The chosen path might violate the total preorder of the metric subset space, thus an exchange of the routing object is necessary.
        2. The node might be overfilled, thus the insertion triggers a split of the node.
Also, at each routing object of the insertion path, we have to assure that the covering radius $r(O_r)$ covers the newly inserted element.
Algorithm~\ref{alg:insert} provides the pseudo code for the \texttt{insert} algorithm.

\begin{algorithm}
    \caption{Insert}
    \begin{lstlisting}
Input: node $N$, object $O$
if $N$ is a leaf node
  insert $O$ to $N$
  return
$O_r\deff \arg\min\left\{ d(O_r,O)\mid O_r\smallerthan O \right\}\cup$
             $\left\{ d(O,O_r)\mid O\smallerthan O_r \right\}$
$insert(T(O_r), O)$
if not $O_r\smallerthan O$
  exchange $O_r$ with $O$
  for each object $O_c$ in $T(O_r)$
    update dist to parent $d(O_r,O_c)$
  $r(O_r) \deff \max_{O_c\in T(O_r)}\left\{ d(O_r,O_c)+r(O_c) \right\}\hspace{-1cm}$
if $split(T(O_r))$
  exchange $O_r$ with promoted objects
  if $N$ is overfilled root node
    add new root node with promoted objects
    \end{lstlisting}
    \label{alg:insert}
\end{algorithm}

\subsubsection{Split Management}

As any other GiST tree, the SuperM-Tree provides a split strategy consisting of a \texttt{Promotion} and \texttt{Partition} algorithm (see Algorithm~\ref{alg:split}).
Spliting a node $N$ creates a new node $N'$.
The \texttt{partition} algorithm partitions the elements of the node $N$ among these two nodes.
We present the \emph{generalized hyperplane} strategy, which puts each object to its nearest neighbor (cf. Algorithm~\ref{alg:partition}).
The \texttt{promote} algorithm provides two routing objects (one for each new partition) which replace the old routing object in the parent node.
If $N$ was the root node, a new node filled with the two promoted routing objects serves as new root node of the tree.

\begin{algorithm}
    \caption{Split}
    \begin{lstlisting}
Input: node $N$, object $O$
if $|N|\leq capacity$
  return false
$(O_1,O_2) \deff promote(N)$
if $promote(N)$ did not find promotion objects:
  return false
$(N_1,N_2) \deff partition(N,O_1,O_2)$
return true, $O_1$, $N_1$, $O_2$, $N_2$
    \end{lstlisting}
    \label{alg:split}
\end{algorithm}
%

A combination of a specific implementation for the \texttt{promote} and \texttt{partition} algorithm is called a split policy.
Although we studied different promotion and partition strategies, we only present the most promising.

\begin{algorithm}
    \caption{Partition}
    \begin{lstlisting}
Input: node $N$, routing objects $O_1,O_2$
 $N_1\deff \left\{ O_j\in N\mid d(O_1,O_j)<d(O_2,O_j) \right\}$
 $N_2\deff N\setminus N_1$
 return $N_1,N_2$
    \end{lstlisting}
    \label{alg:partition}
\end{algorithm}

\subsubsection{Promotion strategies}

This subsection explains the details of our promotion strategy, shown in Algorithm~\ref{alg:promote}.
Note, although the implementation calls the \texttt{partition} algorithm multiple times, our final implementation in C++ actually merges both functions.

The motivation for our promotion strategy is based on the following fact:
Minimizing the overlap of covering areas in a node turned out to be crucial for the performance of multidimensional index structures including the R$^*$-Tree \cite{RStarTree} and the M-Tree \cite{mtree}.
Considering a node $N$ with two routing objects $S$ and $T$ in an M-Tree, the distance between both covering areas is $d(S,T)-r(S)-r(T)$.
When this value becomes negative, both covering areas overlap.
The overlap gets larger when this value shrinks.
Furthermore, the probability increases that a query overlaps both covering areas.
Minimizing the overlap decreases the probability of having to descend both paths.

Since we don't have the concept of \emph{volume of overlap} in metric spaces, we consider the negative of this value as overlap, i.\,e. $overlap(S,T)=r(S)+r(T)-d(S,T)$.

%
\paragraph{Virtual distance and overlap:}
Since a metric subset distance function is not symmetric, the concept of overlap is not transferable to metric subset spaces directly:
Consider the sketched time series $S,T,U$ in Figure~\ref{fig:example_vdist}.
While $T$ and $U$ are not similar at all (and thus might yield no overlap), $S$ had a small distance to both routing objects if they became the promoted objects.
In this case, a query for a time series $Q$ ``containing'' $S$ would need to traverse both paths through the new routing objects $T$ and $U$.
\begin{figure}
    \begin{center}
        \includegraphics[width=.5\linewidth]{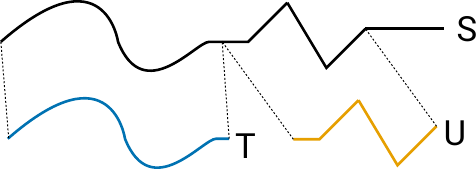}
    \end{center}
    \caption{Example for two very distinct time series ($T$ and $U$) being similar to a third one ($S$).}
    \label{fig:example_vdist}
\end{figure}

Motivated by this example, we introduce the concept of virtual distances regarding a set of third party objects.
\begin{definition}
    Let $\mathcal{R}$ be a set of metric subset objects and $S,T$ be two metric subset objects.
    Then 
    \begin{align*}
        v_{\mathcal R}(S,T) &\deff \min_{R\in\mathcal R, S,T\smallerthan R} \left\{ d(T,R)+d(S,R) \right\}
    \end{align*}
    is called the virtual distance of $S$ and $T$.
\end{definition}
Choosing two routing objects $S$ and $T$ with large virtual distance $v_{\mathcal R}(S,T)$ decreases the probability that a query object $Q$ is similar to objects from both sub trees.
This approach follows the idea of minimizing the overlap.
Thus, we promote a pair of objects minimizing the \emph{virtual overlap} $r(S)+r(T)-v_{\mathcal R}(S,T)$ where $r(S)$ and $r(T)$ are the radii after partitioning and $\mathcal R$ is the set of (routing) objects within the node (cf. Line~$11$).

\begin{algorithm}
    \caption{Promote}
    \begin{lstlisting}
Input: node $N$
$C\deff \left\{ O\mid\forall O_i\in N: O\smallerthan O_i \right\}$
if $|C|=1$ // make sure $|C|\geq 2$
  $C\deff C\cup\left\{ O\mid\forall O_j\in N\setminus C: I\smallerthan O_j \right\}$
$\phi\deff\infty$
for each pair $O_a,O_b\in C$
  $N_a,N_b\deff partition(N,O_a,O_b)$
  if large nodes allowed and
       $|N_a|\leq 1$ or $|N_b|\leq 1$
    skip and ignore $O_a,O_b$
  $p=r(O_a)+r(O_b)-v_N(O_a,O_b)$
  if $p< \phi$
    $O_1\deff O_a;$ $O_2\deff O_b;$ $N_1\deff N_a;$ $N_2\deff N_b$
return $O_1,O_2$
    \end{lstlisting}
    \label{alg:promote}
\end{algorithm}

\paragraph{Strict order in insertion paths:}

The SuperM-Tree maintains the strict order $T\smallerthan P$ for objects $T$ with parent $P$.
In order to respect this property, the promotion strategy only chooses two of the smallest elements of a node (cf. Line~$6$).

\subsubsection{Partition strategy}
\label{sec:partitioningstrategy}

Although numerous possibilities exist for designing partitioning stratregies, we chose the geometric approach for a simple reason:
The trees built with the geometric approach (cf. \emph{generalized hyperplane} in \cite{mtree}) are superior in query performance compared to trees built with the balanced strategy (i.\,e. spliting to partitions of equal size).
Since this insight is the same as in \cite{mtree}, we omit presenting experimental results.

However, we observed the following effect in our examples:
Given a very small object (e.\,g. a time series of length 1) and a larger object, the first one has a lower expected distance to a third object than the second one.
When promoting two objects, one of them is usually smaller than the other.
Hence, most other objects of the node are more similar to the smaller object and the partitioning became more and more unbalanced.
At some point, partitions even degraded, such that only one object (the larger promoted object) is split from the rest of the node.
The resulting tree degraded to a large trunk with lots of very thin branches resulting in bad query performance.

We could prevent this behaviour by ignoring the maximum capacity of a node:
Instead of strictly splitting nodes which exceed the capacity, we look at the resulting partitions and only perform the split, if the partitions are not degenerated (cf. Line~$8$).
In Section~\ref{sec:evaluation}, we show that this strategy outperforms the strict split execution.

\section{Evaluation}
\label{sec:evaluation}

In this section we provide experimental results on the performance in building and querying the SuperM-Tree.

A second goal is to evaluate the flexibility of the concept of metric subset spaces.
Therefore we tested three modular metric subset distance functions, namely the L2~distance on windows, the Dog-Keeper~distance on subsequences, and the Hausdorff~distance on subsets.
For each application we compared both split policies described in Section~\ref{sec:buildingmtree} against a linear scan query algorithm.

\paragraph*{Runtime comparisons:}
In order to understand the influence of different properties of the datasets as well as parameters of the tree, most experiments were based on synthetic datasets.
For the sequence based distance functions ($d_2$ and \dk), we generated random sequences of uniformly distributed lengths between $1$ and $128$.
For the set based distance function, we generated random sets of uniformly distributed cardinality between $1$ and $32$.
We varied the dataset size from $2^8$ to $2^{23}$.
The elements of all sequences and sets are uniformly distributed over a fixed finite interval.
Results regarding query performance are averaged over 100 queries.
Since the node capacity had no big impact on the performance (i.\,e. runtimes were stable for capacities between $64$ and $512$), we fixed the capacity to $128$.
Figure~\ref{fig:timeinsert} shows the time needed for building the tree with successive insert operations and Figure~\ref{fig:time} shows the average query times.
Hereby, the linear scan has been improved with lower bounds already (cf. \cite{ImprovedDK}).

\begin{figure}
    \centering
    \includegraphics[width=.75\linewidth]{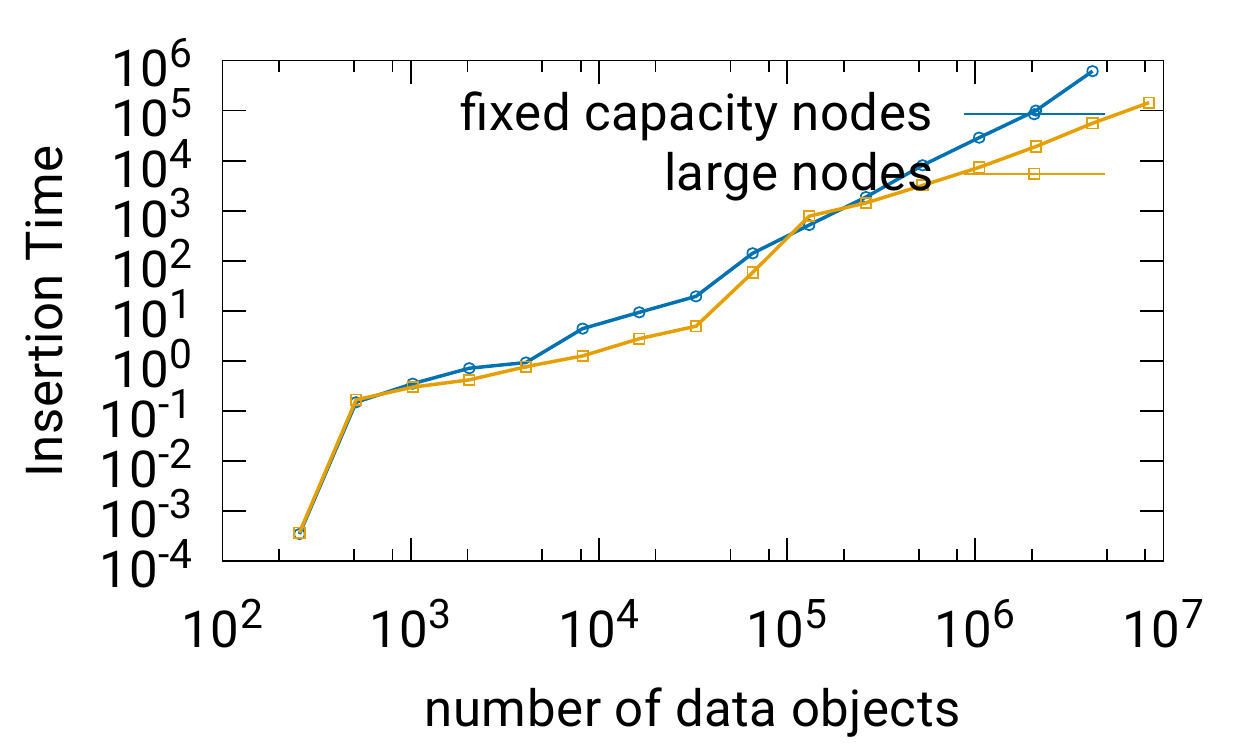}
    \includegraphics[width=.75\linewidth]{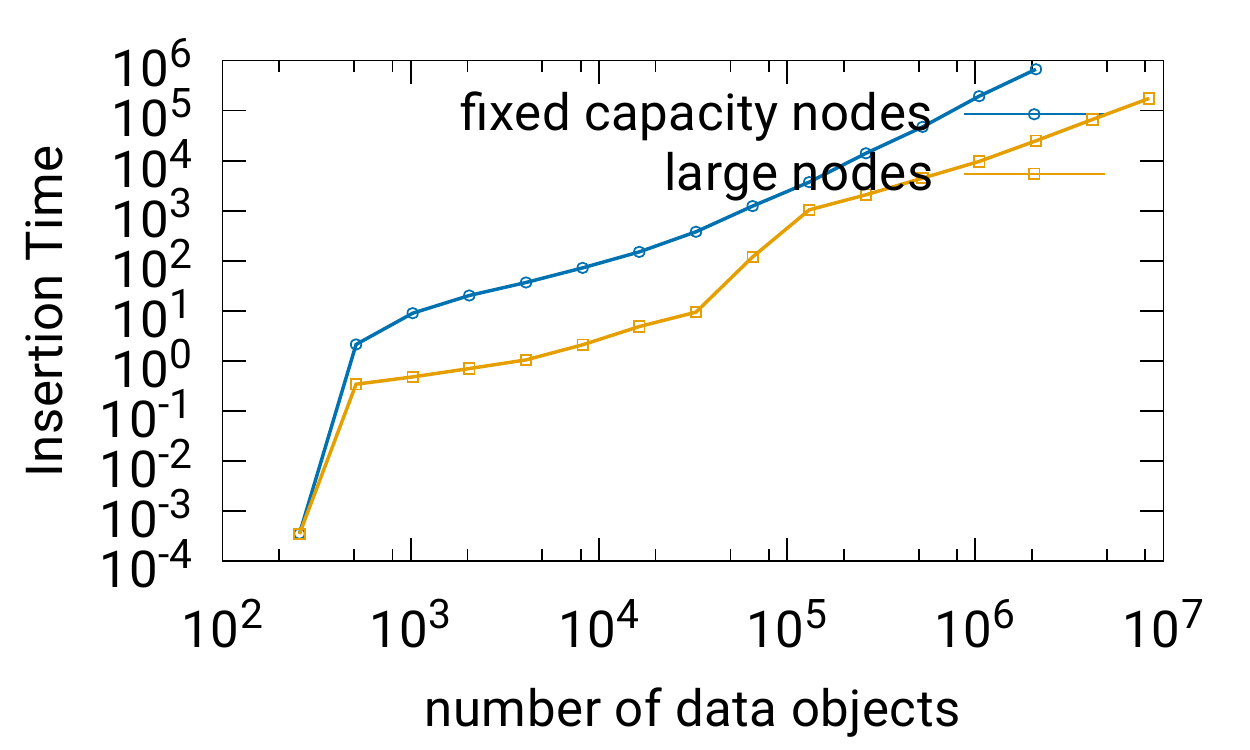}
    \includegraphics[width=.75\linewidth]{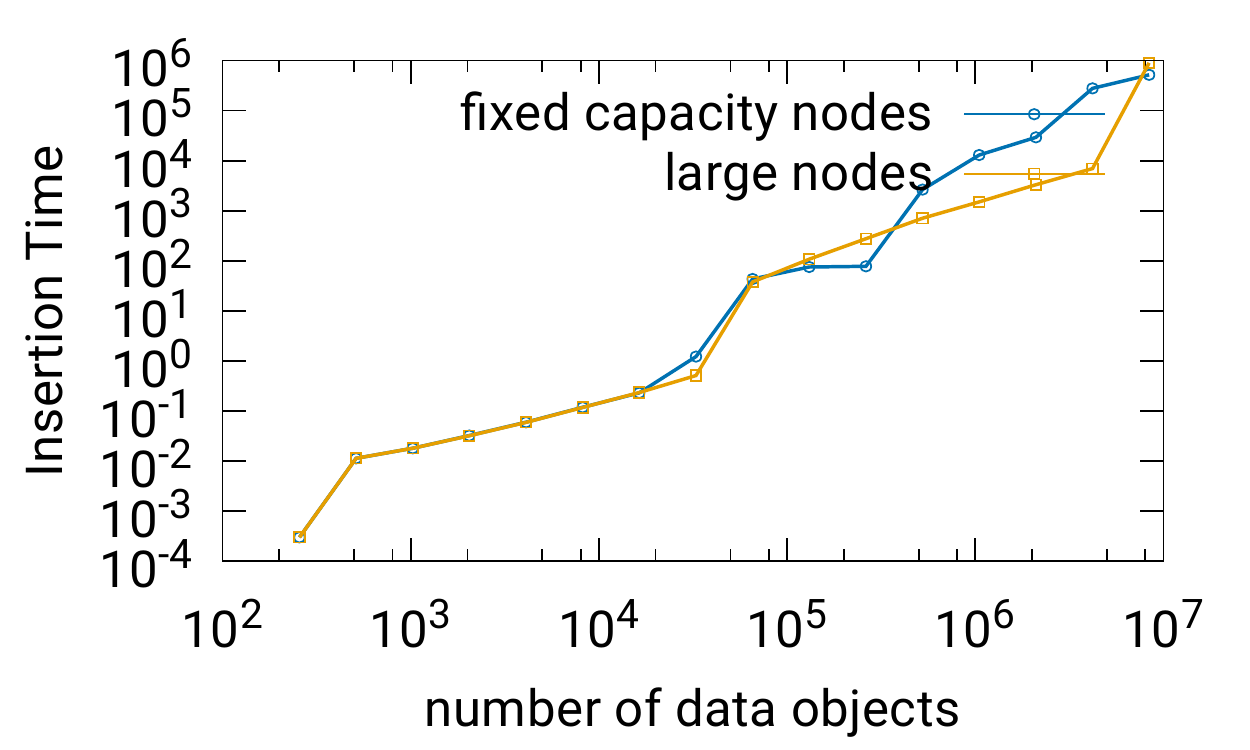}
    \caption{Building time for synthetic data in seconds (top: L2, center: \sdk, bottom: \shd).}
    \label{fig:timeinsert}
\end{figure}

\begin{figure}
    \centering
    \includegraphics[width=.75\linewidth]{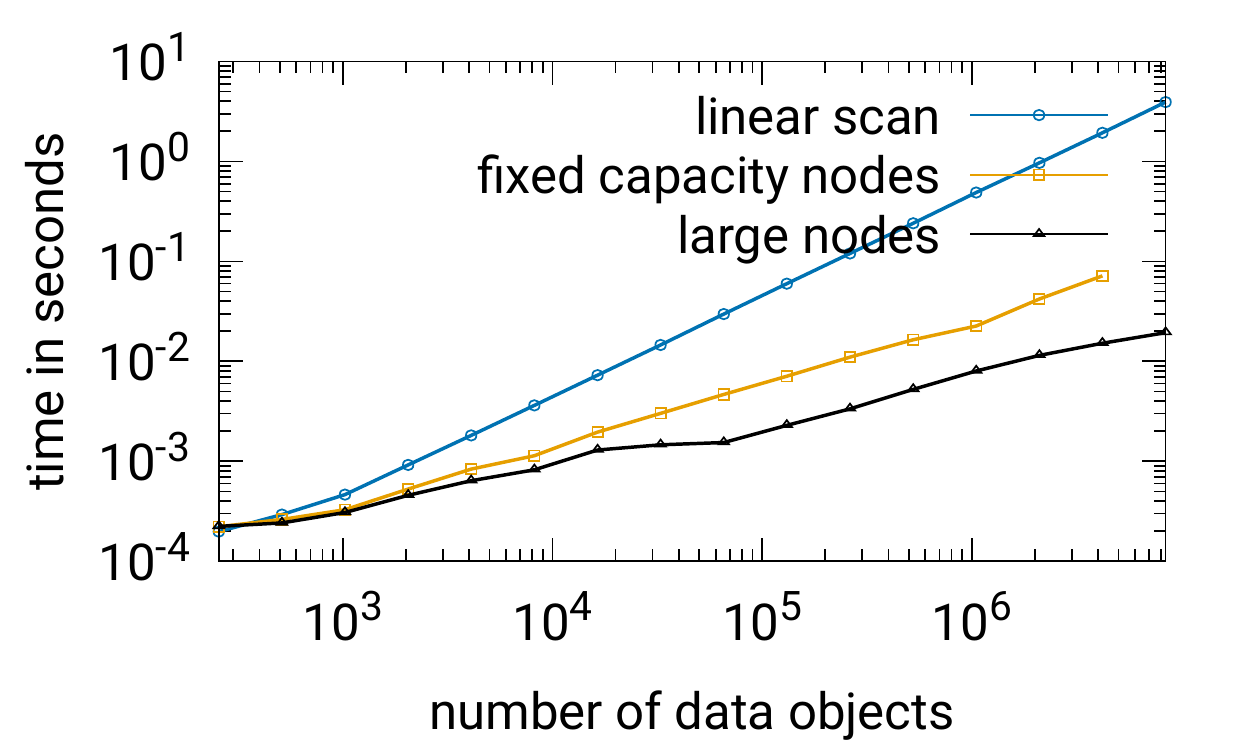}
    \includegraphics[width=.75\linewidth]{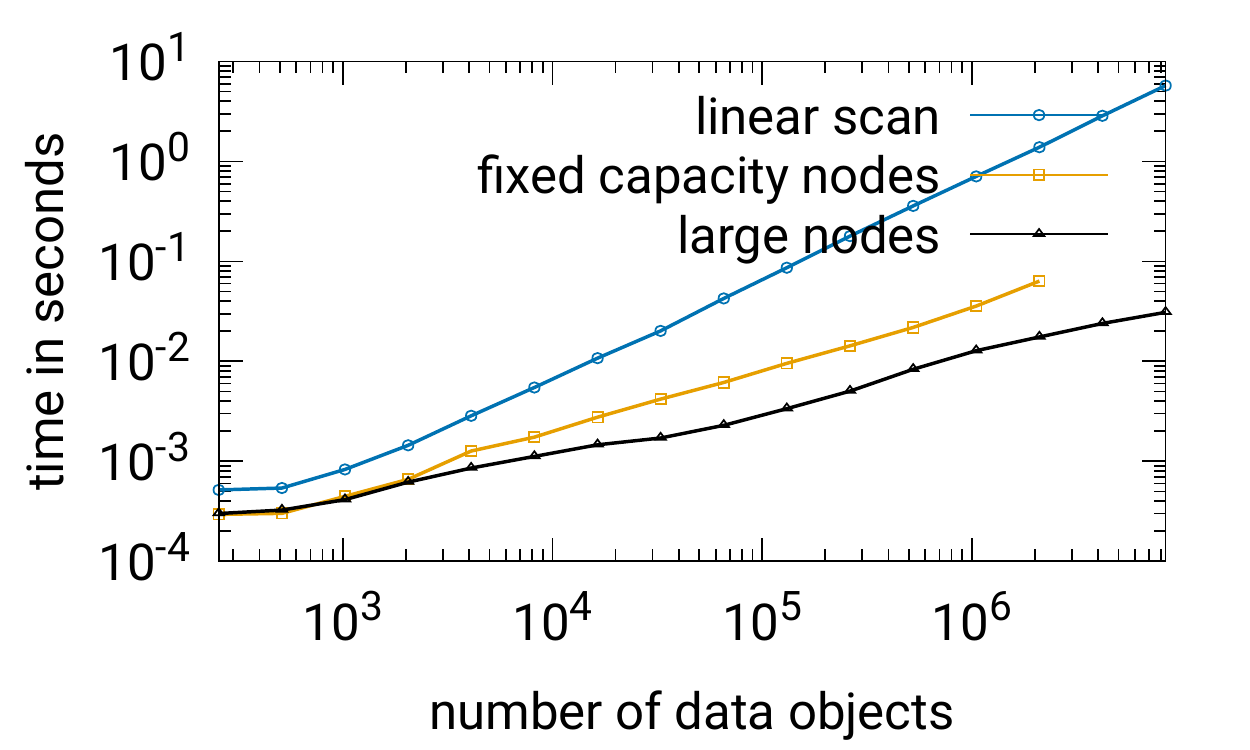}
    \includegraphics[width=.75\linewidth]{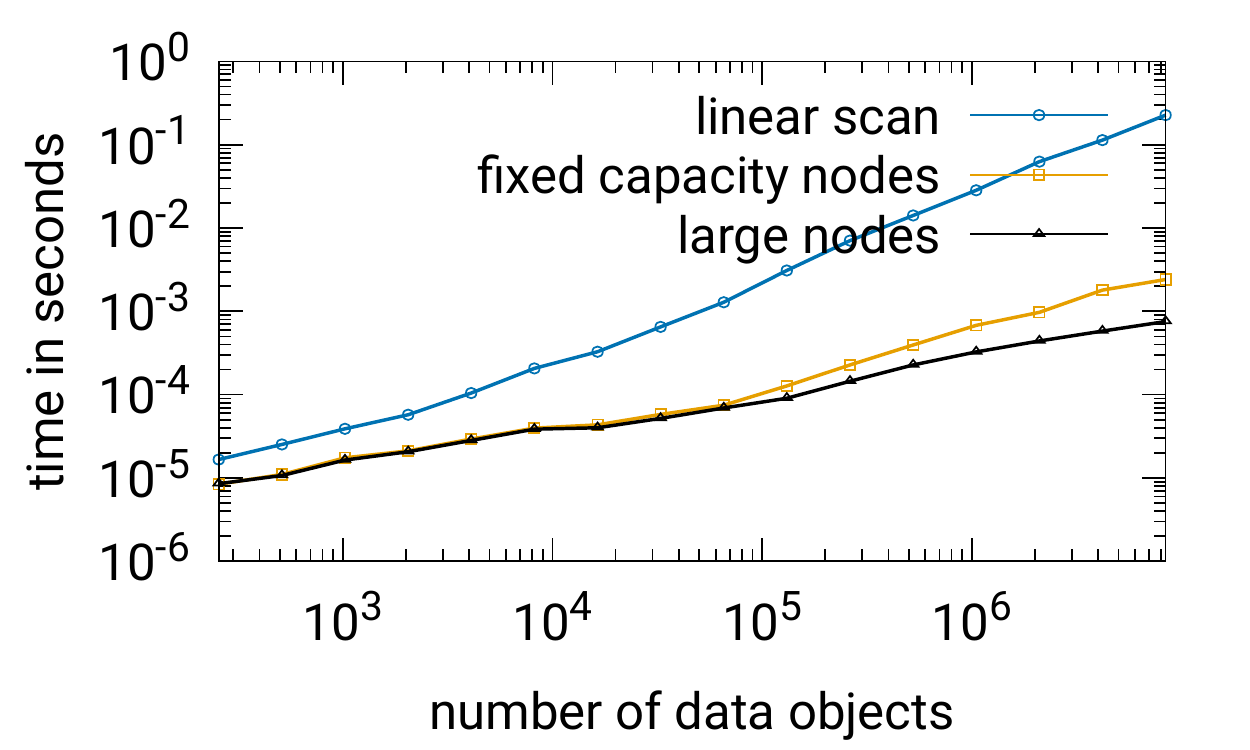}
    \caption{Average query time of 100 1-NN queries on synthetic data(top: L2, center: \sdk, bottom: \shd).}
    \label{fig:time}
\end{figure}

We ran the same experiments for bot sequence based applications on the well known UCR time series benchmark dataset \cite{UCIDatasets}.
We built the trees with each training set respectively and averaged the query times for all samples from the test training set.
To achieve similar properties as in the synthetic datasets, we cropped the training sequences down to random subsequences with a uniformly distributed length between $1$ and $128$.
Figure~\ref{fig:speedup_realworld} shows that the speedup follows the same trend as in the experiments with the synthetic datasets.

\begin{figure}
    \begin{center}
        \includegraphics[width=.75\linewidth]{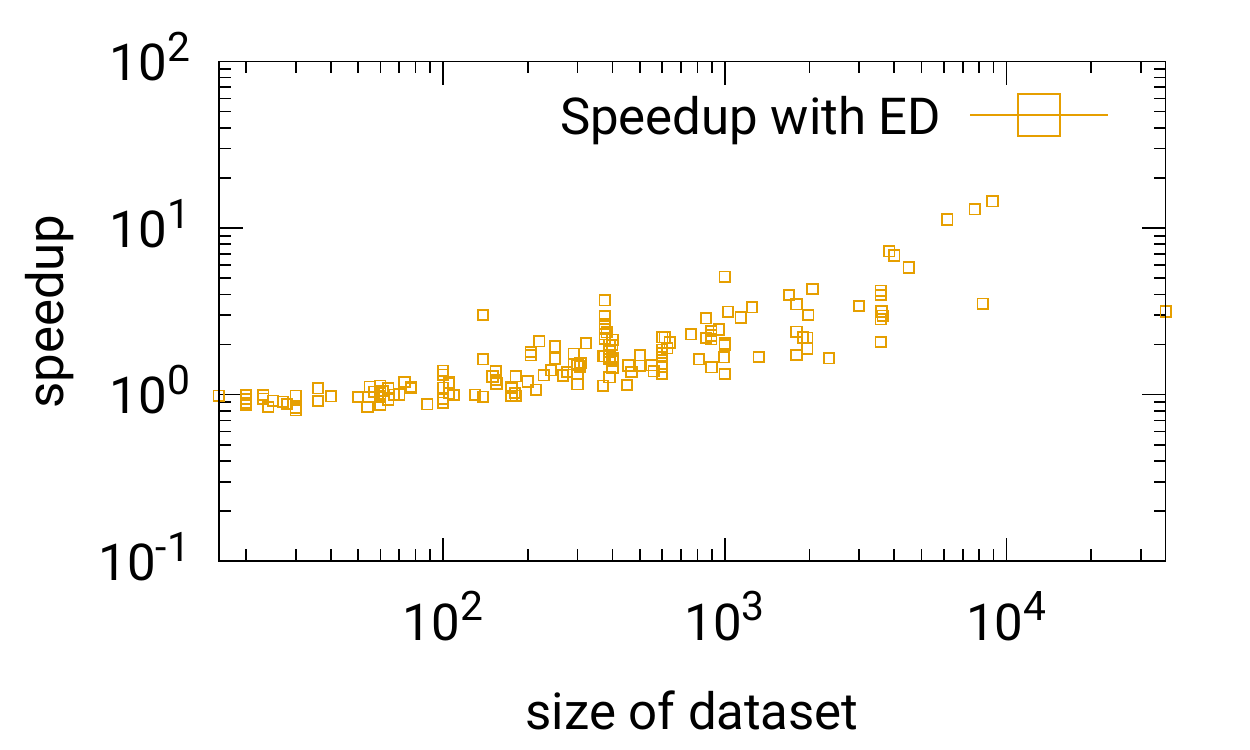}
        \includegraphics[width=.75\linewidth]{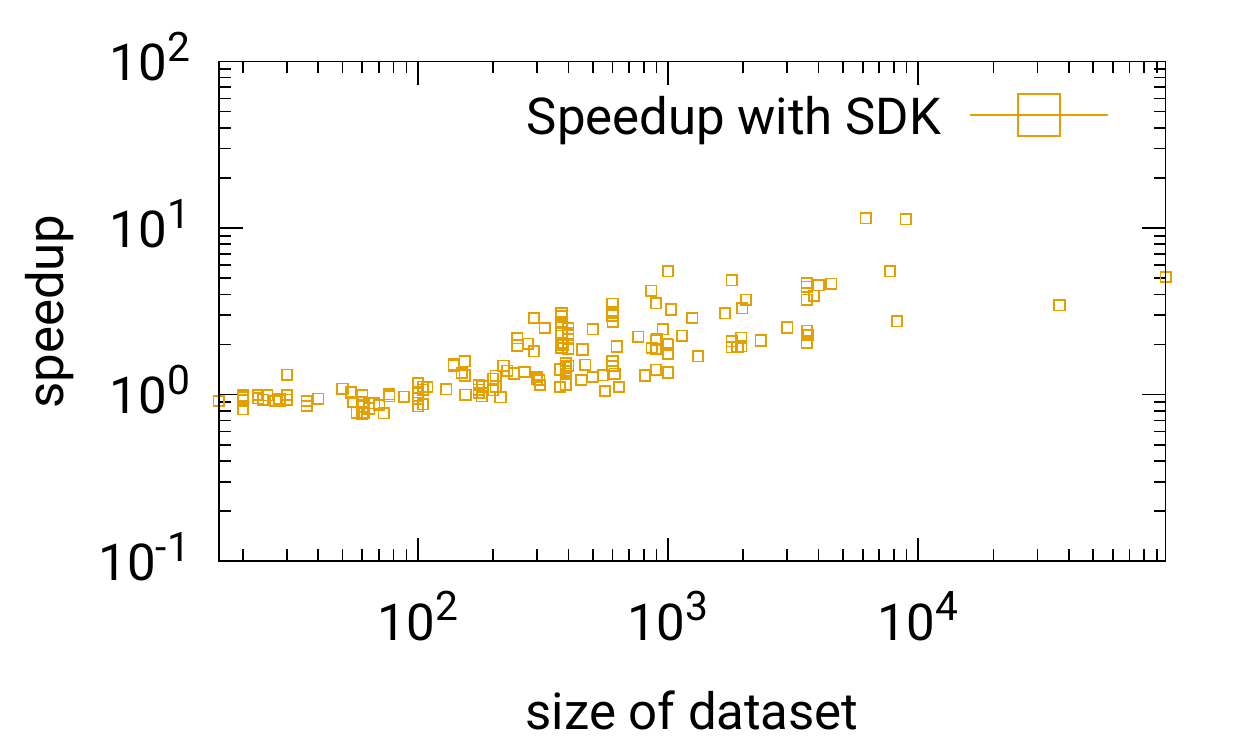}
    \end{center}
    \caption{Average Speedup against linear scan on UCR datasets with L2 (top) and $\sdk$ (bottom)}
    \label{fig:speedup_realworld}
\end{figure}

\paragraph*{Fixed capacity vs. large nodes:}
Figures\ref{fig:timeinsert} shows that the dynamic capacity split policy (\emph{large nodes}) outperforms the static capacity split policy (\emph{fixed capacity}) by more than one order of magnitude while building the tree.
By keeping track of the size of the partitions with the fixed capacity strategy, we could observe a drastic fall down towards a mean of $1$ for small datasets already (i.\,e. a few thousand objects).
This observation confirmed our assumption that the tree degenerates as explained in Section~\ref{sec:partitioningstrategy}.

\paragraph*{The effect of dimensionality:}
The \emph{curse of dimensionality} is a well studied effect and occurs in the SuperM-Tree as well.
The effect appears as seen in all other multi-dimensional index structures.
Due to space limitations we omit presenting experimental results.
TODO: reference!

\paragraph*{Variance of object's size:}
As mentioned earlier, we cropped the sequences down to subsequences of random lengths.
We could observe, that there was no speedup against the linear scan at all if we did not do that.
Hence, the performance strongly depends on having small elements in the dataset.

While spliting overfilled nodes, the smaller elements are being promoted as routing objects.
In the higher levels of the tree (closer to the root), these elements form a space with lower dimensionality.
Thus, pruning works better in the higher levels already and the overall number of pruned elements increases.
On the other hand, datasets with only large objects form a high dimensional space in the higher levels already and thus the curse of dimensionality affects the query performance.

Extending the metric subset space with a function generating smaller (promotion) objects from a given set of objects could solve that problem.
Alternatively, small random elements could be inserted as meta elements manually before or during the whole insertion process.
However, this opens new research areas for each metric subset space respectively.

\section{Conclusion}

We introduced metric subset spaces as a semantic extension of metric spaces.
Three applications from different fields show the flexibility of this concept.

We proposed the SuperM-Tree, a data structure for metric subset spaces derived from the M-Tree.
The SuperM-Tree provides nearest neighbor queries searching for \emph{subobjects} (e.\,g. subsequences or subsets).
Our experiments show, that the index structures outperform the linear scan by multiple orders of magnitude, growing with the size of the dataset.
However, we also observed, that we need small elements in the dataset which act as low dimensional routing objects.
These small elements could be inserted as meta elements manually or generated in a new promotion strategy depending on the specific metric subset space.


\bibliographystyle{plain}
\bibliography{literature}

\end{document}